\documentclass[12pt,reqno]{amsart}
\usepackage{amscd,amssymb,amsmath,amsthm}
\usepackage[arrow,matrix]{xy}
\usepackage{tikz}
\usepackage{graphicx}
\usepackage{cite}

\newtheorem{thm}[subsection]{Theorem}
\newtheorem{lemma}[subsection]{Lemma}
\newtheorem{pro}[subsection]{Proposition}
\newtheorem{cor}[subsection]{Corollary}

\newtheorem{rk}[subsection]{Remark}

\numberwithin{equation}{section} 

\newcommand{\bea}{\begin{eqnarray}}
\newcommand{\eea}{\end{eqnarray}}

\begin{document}

\title[$p$-adic Ising model]
{On a $p$-Adic Generalized Gibbs Measure for Ising Model
on a Cayley Tree}

\author{ Muzaffar Rahmatullaev}
\address{ Muzaffar Rahmatullaev\\
Namangan State University, Namangan, Uzbekistan; Institute of
mathematics, Tashkent, Uzbekistan} \email {{\tt
mrahmatullaev@rambler.ru}}

\author{ Otabek Khakimov}
\address{ Otabek Khakimov\\
Institute of mathematics, Tashkent, Uzbekistan} \email {{\tt
hakimovo@mail.ru}}

\author{Akbarxoja Tukhtaboev}
\address{Akbarxoja Tukhtaboev\\
Namangan Construction Institute, Namangan, Uzbekistan} \email
{{\tt akbarxoja.toxtaboyev@mail.ru}}

\begin{abstract}
In this paper we consider a $p$-adic Ising model on the Cayley tree of order $k\geq 2$. We give full
description of all $p$-adic translation-invariant generalized
Gibbs measures for $k=3$.
Moreover, we show the existence of phase transition for $p$-adic Ising model for any $k\geq3$ when $p\equiv1(\operatorname{mod }4)$.
\end{abstract}

\maketitle

\section{Introduction}

One of the central problems in statistical mechanics
is the study of infinite-volume Gibbs measures corresponding to a
given Hamiltonian. This problem includes the study of phase
transitions: it occurs for a Hamiltonian if there exist at least
two distinct Gibbs measures.
However, a complete analysis of the set of all Gibbs measures for
a specific Hamiltonian is often a difficult problem. For this
reason, most of the work on this topic is devoted to the study of
Gibbs measures on Cayley trees \cite{2}, \cite{6}.

It is known \cite{10}, \cite{24} that $p$-adic models in physics
can not be described using the usual probability theory.
In \cite{10} an abstract $p$-adic probability theory was developed by
the theory of non-Archimedean measures. Probabilistic processes in
the field of $p$-adic numbers have been studied by many authors
(see \cite{1}, \cite{15}, \cite{22}). In \cite{4} non-Archimedean
analogue of Kolmogorov's theorem is proved.

We note that $p$-adic Gibbs measures were studied for
several $p$-adic models of statistical mechanics
\cite{GRR, 7,9,11,12,14,17,18,19,20,Kh1,25,Kh2,MKh}. In \cite{8} it has been
proven that for a $\lambda $ -model on a Cayley tree of order $k$
there is no phase transitions. Recall \cite{9} that the Ising
model is a special case of the $\lambda$ -model.


In \cite{26}, \cite{RR} authors constructed new set real values Gibbs measures for the
Ising model on the Cayley tree. In \cite{27} for the Ising model
on the Cayley tree of order $k$, a new sets of non
translation-invariant ($(k_0)-$translation-invariant) Gibbs
measures is constructed.

In this paper we study existence phase transition for the $p$-adic Ising model on the Cayley tree of order three.
Using the results of \cite{25} by ART construction we give new $p$-adic generalized Gibbs measures for the Ising model
and show the existence of phase transition for arbitrary $k\geq3$.

\section{Preliminaries}

\subsection{$p$-adic Numbers and Measures}

Let $\mathbb Q$
 be the field of rational numbers. For a fixed prime number $p$,
 every rational number $ x\neq 0$ can be represented in the form
$x = {p^r}\frac{n}{m}$ where, $r,n \in {\Bbb Z}$, $m$ is a
positive integer, and $n$ and $m$ are relatively prime with $p$. The $p$-adic norm of $x$ is given by
$$
|x|_p = \left\{ \begin{array}{ll}
p^{-r},& \mbox{if}\ \,x \ne 0,\\
0, & \mbox{if}\ \, x = 0.
\end{array} \right.
$$
This norm is non-Archimedean, i.e. it satisfies the strong
triangle inequality $|x + y|_p\leq\max \{ |x|_p,|y|_p\}$ for all $x,y\in\mathbb Q$.
From this property immediately follow the following facts:

1) if $|x{|_p} \ne |y{|_p}$, then $|x - y{|_p} =\max\{|x{|_p},|y{|_p}\};$

2) if $|x{|_p} = |y{|_p}$, then $|x - y{|_p} = |x{|_p}$.

The completion of $\mathbb Q$ with respect to the $p$-adic norm
defines the $p$-adic field $\mathbb Q_p$  (see \cite{13}).\\The
completion of the field of rational numbers $\mathbb Q$ is either
the field of real numbers $\mathbb R$ or one of the fields of $p$-adic numbers $\mathbb Q_p$ (Ostrowski's theorem).

Any $p$-adic number $x \ne 0$ can be uniquely represented in the
canonical form
$$ x = {p^{\gamma (x)}}({x_0} + {x_1}p + {x_2}{p^2} + ...)
\eqno(2.1)$$ where $ \gamma  = \gamma (x) \in {\Bbb Z}$ and the
integers ${x_j}$ satisfy: ${x_0} > 0,\,\,0 \le {x_j} \le p - 1$ (
see \cite{13}, \cite{23}, \cite{24}). In this case $|x{|_p}\, =
{p^{ - \gamma (x)}}$.

\begin{thm}\label{thmx2=a} \cite{24} The equation ${x^2} = a,\,\,0 \ne a
= {p^{\gamma (a)}}({a_0} + {a_1}p + {a_2}{p^2} + ...),\,\,0 \le
{a_j} \le p - 1,\,{a_0} > 0$ has a solution in $x \in\mathbb Q_p$ iff hold true the following: \\
i) $ \gamma (a)$ is even;\\
ii) ${x^2} \equiv {a_0}(\mbox{mod}\, p)$ is solvable for $p \ne
2$; the equality ${a_1} = {a_2} = 0$ hold if $p = 2$.
\end{thm}
\begin{cor}\label{corx2=-1} \cite{24} The equation ${x^2} =-1$ has a
solution in $\mathbb Q_p$, if and only if
$p\equiv1(\operatorname{mod }4)$.
\end{cor}
For $a \in\mathbb Q_p$ and $r > 0$ we denote
$$
B(a,\,\,r) = \{ x
\in\mathbb Q_p:\,\,|x - a{|_p}\leq r\}.
$$
$p$-Adic \emph{logarithm} is defined by the series
$$\mbox{log}{_p}(x) = \log_p\left( {1 + (x - 1)} \right) = \sum\limits_{n = 1}^\infty  {{{( - 1)}^{n + 1}}\frac{{{{(x - 1)}^n}}}{{n!}}}, $$
which converges for $x \in B(1,1)$ and $p$-adic \emph{exponential}
is defined by $${\exp _p}(x) = \sum\limits_{n = 0}^\infty
{\frac{{{x^n}}}{{n!}}},$$
which converges for $x \in B\left( {0,{p^{ - 1/(p - 1)}}} \right).$\\

We set
$$
\mathcal E_p=\left\{x\in\mathbb Q_p: |x-1|_p<p^{-1/(p-1)}\right\}.
$$
This set is the range of the $p$-adic exponential function. It is
known \cite{Kh1} the following fact.
\begin{lemma}\label{epproperty}Let $p\geq3$.
Then the set $\mathcal E_p$ has the following properties:\\
$(a)$ $\mathcal E_p$ is a group under multiplication;\\
$(b)$ $|a-b|_p<1$ for all $a,b\in\mathcal E_p$;\\
$(c)$ $|a+b|_p=1$ for all $a,b\in\mathcal E_p$;\\
$(d)$ If $a\in\mathcal E_p$, then
there is an element $h\in B_{p^{-1/(p-1)}}(0)$ such that
$a=\exp_p(h)$.
\end{lemma}

A more detailed description of $p$-adic calculus and $p$-adic
mathematical physics can be found in \cite{13}, \cite{23},
\cite{24}.

Let $(X,\mathcal B)$ be a measurable space, where $\mathcal B$ is an
algebra of subsets $X$. A function
$\mu :\mathcal B\to\mathbb Q_p$ is said to be a $p$-adic measure
if for any $ {A_1},{A_2},...,{A_n} \in\mathcal B$ such that\\
${A_i} \cap {A_j} = \emptyset ,\,\,i \ne j$, the following holds:
$$\mu \left( {\bigcup\limits_{j = 1}^n {{A_j}} } \right) = \sum\limits_{j = 1}^n {\mu ({A_j})}.$$

A $p$-adic measure $\mu$ is called \emph{bounded} if $\sup \{
|\mu (A){|_p}:A \in\mathcal B\}  < \infty $ (see, \cite{10}).
It is said that $p$-adic measure is probabilistic if $\mu(X) = 1$ \cite{4}.

\subsection{Cayley tree}

The Cayley tree $\Gamma ^{k}$ of order $k\geq1$ is an infinite
tree i.e., a graph without cycles, such that exactly $k+1$ edges
originate from each vertex. Denote by $V$ the set of vertices, and
by $L$ the set of edges of the Cayley tree $\Gamma ^{k}$. Two
vertices $x$ and $y$ are called \emph{nearest neighbours} if there
exist an edge $l \in L$ connecting them and denote by $l =
\langle x,y\rangle.$

Fix ${x_0} \in {\Gamma ^k}$ and given vertex ${x}$, denote by
$|x|$ the number of edges in the shortest path connecting ${x_0}$
and $x$. For ${x,y}\in{\Gamma^k}$, denote by $d(x,y)$ the number of edges
in the shortest path connecting $x$ and $y$.
For ${x,y}\in{\Gamma^k}$,we write $x\leq y$ if $x$ belongs to the
shortest path connecting ${x_0}$ with y, and we write $x<y$ if
$x\leq y$ and $x\neq y.$ If $x\leq y$ and $|y|=|x|+1$, then we
write $x\rightarrow y$. We call vertex ${x_0}$ the \emph{root} of
the Cayley tree ${x,y}\in{\Gamma^k}$

We set
$$
W_n = \{ x \in V: |x| = n\} ,\ \ {V_n} = \{ x \in V: |x| \le n\},\ \ {L_n} = \{ l =  < x,y >  \in L: x,y \in {V_n}\}
$$
$$
S(x)=\{y\in V: x\rightarrow y\}, \ \ {S_1}(x) = \{ y \in
V: d(x,y) = 1\}.
$$
The set $S(x)$ is called the set of direct successors of the
vertex $x$.

\section{Construction of $p$-adic Gibbs measures for the Ising model}

We consider $p$-adic Ising model on the Cayley tree ${\Gamma
^k}$. Let $\mathbb Q_p$ be a field of $p$-adic numbers and $\Phi=\{-1,1\}$. A
configuration $\sigma$ on V is define  by the function $x \in V
\to \sigma (x) \in \Phi $. Similarly one can define the
configuration ${\sigma _n}$ and ${\sigma ^{_{(n)}}}$ on ${V_n}$
and ${W_n}$, respectively. The set of all configurations on V
(resp.${V_n}$,${W_n}$ ) is denoted by $\Omega  = {\Phi ^V}$ (resp.
${\Omega _{{V_n}}} = {\Phi ^{{V_n}}}$,${\Omega _{{W_n}}} = {\Phi
^{{W_n}}}$ ).

 For given configurations ${\sigma _{n - 1}} \in
{\Omega _{{V_{n-1}}}}$ and ${\varphi ^{(n)}} \in {\Omega _{{W_n}}}$ we
define a configuration in $\Omega _{{V_n}}$ as follows
$$
({\sigma _{n - 1}} \vee {\varphi ^{(n)}})(x) = \left\{
\begin{array}{ll}
{\sigma _{n - 1}}(x), & \mbox{if}\,\,x \in {V_{n - 1}}.\\
{\varphi ^{(n)}}(x), & \mbox{if}\,\,x \in {W_n}.
\end{array} \right.
$$

A formal $p$-adic Hamiltonian $H:{\Omega} \to\mathbb Q_p$ of
Ising model is defined as
\begin{equation}\label{Ham}
H(\sigma ) = J\sum\limits_{  \langle x,y\rangle \in L} {\sigma
(x)\sigma (y)},
\end{equation}
where $|J|_p < p^{ - 1/(p - 1)}$ for any $ < x,y >\in L$.

Let $h:x \to {h_x} \in\mathbb Q_p\setminus\{0\}$ be a $p$-adic function on
$V$. Consider $p$-adic probability distribution $\mu _h^{(n)}$ on
${\Omega _{{V_n}}}$, which is defined as
\begin{equation}\label{mu}
\mu _h^{(n)}(\sigma_n) = Z_{n,h}^{- 1}\exp _p\{ {H_n}(\sigma _n)\}
\prod\limits_{x \in W_n}h_x^{(\sigma _n)},\ \ n =
1,2,\dots,
\end{equation}
where ${Z_{n,h}}$ is the normalizing constant
\begin{equation}\label{Z}
Z_{n,h}=\sum\limits_{\varphi\in{\Omega _{{V_n}}}}\exp_p\{
H_n(\varphi )\}\prod\limits_{x \in W_n}h_x^{\varphi (x)},
\end{equation}

$${H_n}(\sigma _n)=J\sum\limits_{ \langle x,y\rangle \in
L_n} {\sigma_n (x)\sigma_n (y)}.$$

A $p$-adic probability distribution
$\mu _h^{(n)}$ is said to be consistent if for all $n\geq1$ and
${\sigma _{n - 1}} \in {\Omega _{{V_{n - 1}}}}$, we have
\begin{equation}\label{cc1}
\sum\limits_{\varphi  \in {\Omega _{{W_n}}}} {\mu
_h^{(n)}({\sigma _{n - 1}} \vee \varphi ) = \mu _h^{(n -
1)}({\sigma _{n - 1}.})}
\end{equation}

In this case, by the $p$-adic analogue of Kolmogorov theorem
\cite{4}, there exists a unique measure ${\mu _h}$ on the set $\Omega $
such that ${\mu _h}\left( {\left\{ {\sigma {|_{{V_n}}}\equiv{\sigma
_n}} \right\}} \right) = \mu _h^{(n - 1)}\,\,({\sigma _{n - 1}})$
for all $n$ and ${\sigma _{n - 1}} \in {\Omega _{{V_{n -
1}}}}.$

The measure ${\mu _h}$ is called $p$-adic generalized Gibbs measure corresponding to the function $h:x \to
{h_x} \in\mathbb Q_p$ if restriction of $\mu_h$ to $V_n$ is a measure \eqref{mu}.
We notice that if $h_x\in\mathcal E_p$ for all $x\in V$ then
corresponding measure is called $p$-adic Gibbs measure. It is said that a phase transition
occurs for a given hamiltonian if there exist at least two measures. Moreover,
if one of them is not bounded and another one is bounded then it is said that
there exists the strong phase transition for that model.

\begin{pro}\label{aspro}\cite{25} A sequence of $p$-adic probability distributions $\mu
_h^{(n)},  n = 1,2,...$ determined by formula (3.2) is
consistent if and only if for any $x \in V\backslash \{ {x_0}\}$,
we have the equality
\begin{equation}\label{exp(h)}
h_x^2=\prod\limits_{y
\in S(x)}\frac{\theta h_y^2+1}{h_y^2+\theta},
\end{equation}
where ${{\theta} = {{\exp }_p}\{ 2{J}\} }$.
\end{pro}
\begin{rk}\label{rem}It is easy to see that if the function ${h_x}$
is a solution to equation \eqref{exp(h)}, then the function ${-h_x}$ is
also a solution.  If  we  consider  Ising model  on  the  Cayley
tree of order $k$,  then  these solutions define  the same
measure ${\mu_h}$.
\end{rk}

\section{$p$-Adic Translation-Invariant Generalized  Gibbs Measure}

We consider $p$-adic translation-invariant generalized  Gibbs
measures for Ising model on the Cayley tree ${\Gamma
^3}$. Thanks to Proposition \ref{aspro}, in order to find all translation-invariant measures
it is enough to consider the following equation
\begin{equation}\label{eq}
h^2 = \left(\frac{\theta h^2 + 1}{h^2 + \theta}\right)^3.
\end{equation}
Denoting  $z={h^2}$, from the last one we get
\begin{equation}\label{eq1}
z = \left(\frac{\theta z + 1}{z + \theta}\right)^3,
\end{equation}
which is equivalent to
\begin{equation}\label{eq2}
(z^2-1)\left(z^2+(3\theta-\theta^3)z+1\right)=0.
\end{equation}
The equation \eqref{eq2} has at least two solutions $z_{1,2}=\pm1$. If there exists
$\sqrt{\theta^2-4}$ in $\mathbb Q_p$ then Eq. \eqref{eq2} has exactly four solutions, which are
$z_{1,2}$ and
$$
z_{3,4}=\frac{\theta ^3-3\theta\pm(\theta ^2-1)\sqrt{\theta ^2 - 4}}{2}.
$$
Denote $\Delta(\theta)=\theta^2-4$.
\begin{lemma}\label{lem1} A number  $\sqrt {\Delta (\theta )}$ exists in
$\mathbb Q_p$ if and only if  $p\equiv1(\operatorname{mod }6)$.
\end{lemma}
\begin{proof}
Since $\theta\in\mathcal E_p$, due to Lemma \ref{epproperty} we
have $\left|\theta^2-1\right|_p<1$. it yields that
$|\Delta(\theta)+3|_p<1$. Then according to Theorem \ref{thmx2=a},
a number $\sqrt{\Delta(\theta)}$ exists in $\mathbb Q_p$ if and
only if $\sqrt{-3}\in\mathbb Q_p$. Again thanks to Theorem
\ref{thmx2=a} existence of $\sqrt{-3}$ is equivalent to the
solvability of $x^2+3\equiv0(\operatorname{mod }p)$. It is known
that (see \label{aspro}\cite{28}) the congruence
$x^2+3\equiv0(\operatorname{mod }p)$ is solvable if and only if
$p\equiv1(\operatorname{mod }6)$.
\end{proof}
Since $z_3z_4=1$ one can conclude that the existence of $\sqrt{z_3}$ implies the
existence of $\sqrt{z_4}$ in $\mathbb Q_p$. For this reason it is enough to check the existence $\sqrt{z_3}$
in to describe the set of all $p$-adic translation-invariant generalized Gibbs measures for homogenous Ising model.

\begin{lemma}\label{lem2}
The number $\sqrt{z_3}$ exists in $\mathbb Q_p$ if and only if $p\equiv1(\operatorname{mod }12)$.
\end{lemma}
\begin{proof} Assume that $z_3$ is a solution of \eqref{eq2}.
Then due to Lemma \ref{lem1} we have $p\equiv1(\operatorname{mod }6)$.
We obtain
\begin{eqnarray*}
z_3+1&=&\frac{\theta^3-3\theta+2+(\theta^2-1)\sqrt{\Delta(\theta)}}{2}\\
&=&\frac{(\theta-1)\left(\theta^2+\theta-2+(\theta+1)\sqrt{\Delta(\theta)}\right)}{2}.
\end{eqnarray*}
Since $p\neq2$ and $|\theta-1|_p<1$ one has $|z_3+1|_p<1$.
Again keeping in mind $p\neq2$ and thanks to Theorem \ref{thmx2=a} the existence of $\sqrt {{z_3}}$ is equivalent
to the existence of the number $\sqrt{-1}$. Then using Corollary \ref{corx2=-1} we conclude that
$\sqrt{z_3}\in\mathbb Q_p$ if and only if $p\equiv1(\operatorname{mod }4)$.
After combining $p\equiv1(\operatorname{mod }6)$ with $p\equiv1(\operatorname{mod }4)$ we get $p\equiv1(\operatorname{mod }12)$,
 which is necessarily and sufficiently condition of the existence $\sqrt{z_3}$ in $\mathbb Q_p$.
\end{proof}
\begin{pro}\label{pro1}
Let $\mathcal N_p$ be a number of solutions of \eqref{eq}. Then
$$
\mathcal N_p=\left\{
\begin{array}{ll}
2, & \mbox{if }\ p\not\equiv1(\operatorname{mod }4)\\
4, & \mbox{if }\ p\equiv1(\operatorname{mod }4), p\not\equiv1(\operatorname{mod }3)\\
8, & \mbox{if }\ p\equiv1(\operatorname{mod }12)
\end{array}
\right.
$$
\end{pro}
\begin{proof}
Assume that $p\not\equiv1(\operatorname{mod }4)$. Then due to Corollary \ref{corx2=-1} we have
$\sqrt{z_2}\not\in\mathbb Q_p$. Moreover, according to Lemma \ref{lem2} one has $\sqrt{z_3},\sqrt{z_4}\not\in\mathbb Q_p$.
Hence, in this case \eqref{eq} has exactly two solutions: $h_1=\sqrt{z_1}$ and $-h_1$. Now, we suppose that
$p\equiv1(\operatorname{mod }4)$ and $p\not\equiv1(\operatorname{mod }3)$. Then again due to Corollary \ref{corx2=-1} there exists
$\sqrt{z_2}$ in $\mathbb Q_p$ and by Lemma \ref{lem2} we get $z_{3,4}\not\in\mathbb Q_p$. Consequently,
\eqref{eq} has exactly four solutions: $\pm h_1$, $h_2=\sqrt{z_2}$ and $-h_2$. Let us assume that
$p\equiv1(\operatorname{mod }12)$. In this case according to Corollary \ref{corx2=-1} and by Lemma \ref{lem2} there exist
$\sqrt{z_2}$ and $\sqrt{z_{3,4}}$ in $\mathbb Q_p$, which imply that Eq. \eqref{eq} has exactly eight solutions:
$\pm h_1, \pm h_2$, $h_3=\sqrt{z_3}$ and $-h_3$.
\end{proof}
We denote by $TIpGGM(H)$ the set of all $p$-adic translation-invariant generalized Gibbs measures for
hamiltonian $H$. Notation $|A|$ means cardinality of the set $A$.

\begin{thm}\label{thmdesc} Let $H$ be an Ising model on a Cayley tree order three. Then it holds the following:
$$
|TIpGGM|=\left\{
\begin{array}{ll}
1, & \mbox{if }\ p\not\equiv1(\operatorname{mod }4)\\
2, & \mbox{if }\ p\equiv1(\operatorname{mod }4), p\equiv2(\operatorname{mod }3)\\
4, & \mbox{if }\ p\equiv1(\operatorname{mod }12)
\end{array}
\right.
$$
\end{thm}
The proof follows from Proposition \ref{aspro}, Remark \ref{rem} and Proposition \ref{pro1}.

Now we study the boundedness of $p$-adic translation-invariant generalized Gibbs measures $\mu_{h_i}$, $i=1,2,3,4$.
We need some auxiliary lemmas.
\begin{lemma}\label{lemrecZ}\cite{25} Let $h$ be translation-invariant
solution to equation \eqref{eq} and ${\mu_h}$ be a corresponding
$p$-adic translation-invariant generalized Gibbs measure. Then
for normalizing constant ${Z_{n,h}}$  holds
\begin{equation}\label{recZ}
Z_{n + 1,h} = A_{n,h}Z_{n,h},
\end{equation}
where
$$
A_{n,h} =\left(\frac{\left(\theta{h^2}+1\right)\left(h^2+\theta\right)}{\theta {h^2}}
\right)^{2^{n - 1}}.$$
\end{lemma}
\begin{lemma}\label{lem3}
The norms of the solutions
${h_1}$, $h_2$, $h_3$, $h_4$ are equal to one.
\end{lemma}

\begin{proof} Since $|z_i|_p=1$, $i=1,2,3,4$ one can immediately calculate
$|\sqrt{z_{i}}|_p=1$, $i=1,2,3,4$.
\end{proof}

\begin{lemma}\label{lem4} For the normalizing constants ${Z_{n + 1,{h_i}}}$, $i = 1,2,3,4$  we have
\begin{enumerate}
\item[(i)] $\left|Z_{n + 1,h_1}\right|_p = \left\{ \begin{array}{ll}
1, & \mbox{if}\ p \neq2,\\
2^{-2^n + 2}, & \mbox{if}\ p = 2.
\end{array} \right.$\\
\item[(ii)] $\left|Z_{n + 1,h_i}\right|_p\leq|\theta-1|_p^{2^n - 2}$, $i = 2,3,4$.
\end{enumerate}
\end{lemma}
\begin{proof} (i) For normalizing constant ${Z_{n,{h_1}}}$  from
\eqref{recZ} we obtain
$$
Z_{n,h_1} = A_{n - 1,h_1}Z_{n - 1,h_1} =\dots= A_{n-1,h_1}A_{n - 2,h_1}\dots{A_{1,{h_1}}}{Z_{1,{h_1}}} =
{\left( {\frac{{{{\left( {\theta  + 1} \right)}^2}}}{\theta }}
\right)^{{2^{n - 1}} - 1}}.
$$
Since $\theta\in\mathcal E _p$
we have  ${\left| {{Z_{n + 1,{h_1}}}} \right|_p} = \left\{
\begin{array}{ll}
1, & \mbox{if}\ p\neq2,\\
{2^{ - {2^n} + 2}}, & \mbox{if}\ p = 2.
\end{array} \right.$\\
(ii) Let ${h_2}$ be a solution to \eqref{eq}. Then due to non-Archimedean norm's property
we have
$$
Z_{n,h_2} = \left(\frac{\left(\theta h_2^2 + 1\right)\left(h_2^2 + \theta\right)}
{\theta h_2^2}\right)^{2^{n - 1} - 1} =\left(
\frac{\left(\theta  - 1\right)^2}{\theta}
\right)^{2^{n - 1} - 1}.
$$
Thus we have ${\left| {{Z_{n + 1,{h_2}}}} \right|_p} = \left|
{\theta  - 1} \right|_p^{{2^n} - 2}.$\\
 (iii) Let ${h_3}$ be a solution to \eqref{eq}.
 \begin{eqnarray*}
 |h_3^2 + 1|_p&=&|z_3 + 1|_p\\
 &=&\left|\frac{\theta ^3 - 3\theta+(\theta^2 - 1)\sqrt{\Delta(\theta)}}{2} + 1
 \right|_p\\
&=&\left|\frac{(\theta  - 1)\left((\theta-1)(\theta+2)+(\theta  + 1)\sqrt{\Delta(\theta)}\right)}{2} \right|_p\\
&=&\left|\theta - 1\right|_p.
\end{eqnarray*}
Then using strong triangle inequality we have
$${Z_{n,{h_3}}} = {\left( {\frac{{\left( {\theta h_3^2 + 1}
\right)\left( {h_3^2 + \theta } \right)}}{{\theta h_3^2}}}
\right)^{{2^{n - 1}} - 1}} = {\left( {\frac{{\left( {\theta ({z_3}
+ 1) + 1 - \theta } \right)\left( {{z_3} + 1 + \theta  - 1}
\right)}}{{\theta {z_3}}}} \right)^{{2^{n - 1}} - 1}}.
$$
Again noting $\theta\in\mathcal E_p$ and ${\left| {{h_3}}
\right|_p} = 1$ (see Lemma \ref{lem4}) we can find
$$
{\left| {{Z_{n + 1,{h_3}}}} \right|_p}
\le \left| {\theta  - 1} \right|_p^{{2^n} - 2}.
$$
Using ${z_3}{z_4} =1$ we obtain the following
\begin{eqnarray*}
Z_{n,{h_4}}& =& {\left( {\frac{{\left( {\theta h_4^2 + 1}
\right)\left( {h_4^2 + \theta } \right)}}{{\theta h_4^2}}}
\right)^{{2^{n - 1}} - 1}}\\
& =& {\left( {\frac{{\left( {\theta {z_4}
+ 1} \right)\left( {{z_4} + \theta } \right)}}{{\theta {z_4}}}}
\right)^{{2^{n - 1}} - 1}}\\
&=& {\left( {\frac{{z_4^2\left( {\theta  + {z_3}} \right)\left( {1 + \theta {z_3}} \right)}}{{\theta {z_4}}}} \right)^{{2^{n - 1}} - 1}}\\
&=& {\left( {\frac{{\left( {\theta  + {z_3}} \right)\left( {1 + \theta {z_3}} \right)}}{{\theta {z_3}}}} \right)^{{2^{n - 1}} - 1}}\\
&=& {\left( {\frac{{\left( {\theta  + h_3^2} \right)\left( {1 + \theta h_3^2} \right)}}{{\theta h_3^2}}} \right)^{{2^{n - 1}} - 1}}\\
& =& {Z_{n,{h_3}}}.
\end{eqnarray*}
Consequently, ${\left| {{Z_{n + 1,{h_4}}}} \right|_p} \le \left|
{\theta  - 1} \right|_p^{{2^n} - 2}.$\\
\end{proof}
\begin{thm}\label{thmbound} Let $H$ be an Ising model on a Cayley tree order three.
Then the following statements are hold:
\begin{enumerate}
\item[(i)] if $p = 2$ then the unique $p$-adic
translation-invariant generalized Gibbs measure ${\mu _{{h_1}}}$ is not bounded.\\
\item[(ii)] If $p\neq 2$ then among the $p$-adic translation-invariant
generalized  Gibbs measures only ${\mu _{{h_1}}}$ is
bounded.
\end{enumerate}
\end{thm}
\begin{proof} (i) Let $p = 2$. In this case due to Theorem \ref{thmdesc}
there exists a unique $p$-adic translation-invariant generalized
Gibbs measure ${\mu _{{h_1}}}$. Then
by Lemma \ref{lem3} and Lemma \ref{lem4} we have
$$
|\mu _{{h_1}}^{(n)}(\sigma ){|_2} =
{2^{{2^n} -
2}},\ \ \ \ \forall n\in\mathbb N,\ \forall\sigma
\in {\Omega _{{V_n}}}.
$$
Hence, $|\mu _{{h_1}}^{(n)}(\sigma ){|_2} \to \infty$ as $n \to \infty$. It means that the measure ${\mu
_{{h_1}}}$ is not bounded.\\
(ii)  Let $p \neq 2$. According to Lemma \ref{lem3} and Lemma \ref{lem4} one has
$$
|\mu _{{h_1}}^{(n)}(\sigma ){|_2} =
1,\ \ \ \ \forall n\in\mathbb N,\ \forall\sigma
\in {\Omega _{{V_n}}},
$$
which implies boundedness of limiting measure $\mu_{h_1}$. Moreover, if there exists
at least one of the measures $\mu_{h_2}$, $\mu_{h_3}$ and $\mu_{h_4}$ then again using
Lemma \ref{lem3} and Lemma \ref{lem4} we can verify that
they are not bounded.
\end{proof}
Thanks to Theorem \ref{thmdesc} and Theorem \ref{thmbound} we get the following
\begin{thm} If $p\equiv1(\operatorname{mod }12)$ then for the Ising model on a Cayley tree
order three the strong phase
transition occurs.
\end{thm}

\section{Existence of phase transition for $p$-adic Ising model on $\Gamma^k$: $k\geq3$}

\subsection{$p$-Adic ART Generalized Gibbs measures}

In \cite{25} $p$-adic translation-invariant and periodic generalized
 Gibbs measures for the Ising model on the Cayley tree of order
two are studied. In the previous section we have shown that the strong phase transition
occurs on a Cayley tree of
order three. In this section we are going to describe new
$p$-adic generalized Gibbs measures of the Ising model on the
Cayley tree of order $k\geq 3$ by method ART (see \cite{26}).

We recall that each solution of \eqref{exp(h)}
define a $p$-adic generalized Gibbs measures for Ising model on the Cayley tree of order $k\geq 1$.
One can see that $h_x=1$ for all $x\in V$ is a solution of \eqref{exp(h)} for any $k\geq1$.
Now we construct new solutions of \eqref{exp(h)} for $k\geq 3$. If $k=2$ than all translation-invariant
solutions of \eqref{exp(h)} can be found from the following
\begin{equation}\label{k=2}
h^2 = \left( \frac{\theta {h^2} + 1}{h^2 + \theta}
\right)^2.
\end{equation}

In \cite{25} have been proved, that the equation \eqref{k=2} has a unique solution if $p\not\equiv 1(\operatorname{mod }4)$
and it has exactly three solution if $p\equiv1(\operatorname{mod }4)$. In what follows we assume that
$p\equiv1(\operatorname{mod }4)$. In this case due to results of \cite{25} the followings
\begin{equation}\label{123}
h^{(2)}_0=1,\ h^{(2)}_{1,2}=\frac{\theta-1\pm\sqrt{(\theta-3)(\theta+1)}}{2}
\end{equation}
are solutions of \eqref{k=2}. For $k\geq3$ we give construction of some solutions of \eqref{exp(h)} using
\eqref{123}.

Let $V^k$ be the set of all vertices of the Cayley tree
$\Gamma^k$. Since $k>2$ one can consider $V^{2}$ as a subset of
$V^k$. Define the following function
\begin{equation}\label{newsol}
\widetilde{h}^{(i)}_x=\left\{%
\begin{array}{ll}
    h^{(2)}_{i}, &\mbox{if} \ \ {x \in V^2}, \\
    1, &\mbox{if} \ \  {x \in V^k\setminus V^2}, \\
\end{array}%
\right.
\end{equation}
where $i=1,2$.  This function on the Cayley tree of order $k=3$
is shown in Fig.\ref{f4}.

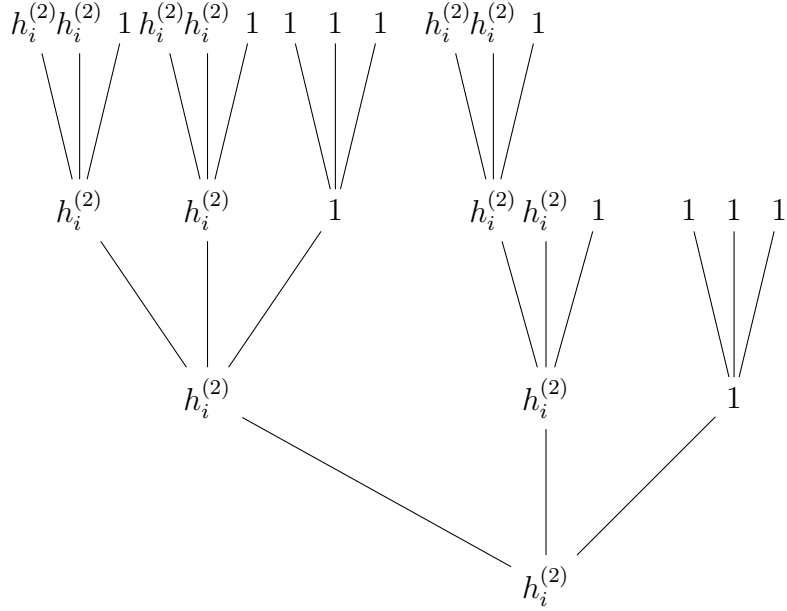
\begin{figure}
\scalebox{1}{
\begin{tikzpicture}[level distance=2.5cm,
level 1/.style={sibling distance=4.5cm},
level 2/.style={sibling distance=1.7cm},
level 3/.style={sibling distance=.6cm}]
\node {$h^{(2)}_{i}$} [grow'=up]
    child {node {$h^{(2)}_{i}$}
        child {node {$h^{(2)}_{i}$}
            child {node {$h^{(2)}_{i}$}}
            child {node {$h^{(2)}_{i}$}}
            child {node {$1$}}
        }
        child {node {$h^{(2)}_{i}$}
            child {node {$h^{(2)}_{i}$}}
            child {node {$h^{(2)}_{i}$}}
            child {node {$1$}}
        }
        child {node {$1$}
            child {node {$1$}}
            child {node {$1$}}
            child {node {$1$}}
        }
    }
    child {node {$h^{(2)}_{i}$}
        child[sibling distance=.7cm] {node {$h^{(2)}_{i}$}
            child[sibling distance=.6cm] foreach \name in {h^{(2)}_{i},h^{(2)}_{i},1} { node
        {$\name$} }
        }
        child[sibling distance=.7cm] {node {$h^{(2)}_{i}$}}
        child[sibling distance=.7cm] {node {$1$}}
    }
    child[sibling distance=2.5cm] {node {$1$}
        child[sibling distance=.6cm] foreach \name in {1,1,1} { node
        {$\name$} }
    }
;
\end{tikzpicture}
}
\caption{\footnotesize \noindent
The function $\widetilde{h}^{{(i)}}_x, i=1,2$ on the Cayley tree of order three.}\label{f4}
\end{figure}

Now we shall check that \eqref{newsol} satisfies \eqref{exp(h)} on $\Gamma^k$.

Let $x\in V^{2}\subset V^k$. For $i=1, 2$ we have
\begin{eqnarray*}
(\widetilde{h}^{(i)}_x)^2 &=& \prod\limits_{y \in S(x)}
{\frac{{{\theta}(\widetilde{h}^{{(i)}}_y)^2 +
1}}{{(\widetilde{h}^{{(i)}}_y)^2 + {\theta}}}}\\
&=&\prod\limits_{y \in
S(x)\cap V^2} {\frac{{{\theta}(\widetilde{h}^{{(i)}}_y)^2 +
1}}{{(\widetilde{h}^{{(i)}}_y)^2 + {\theta}}}}\times\prod\limits_{y
\in S(x)\cap (V^k\setminus V^2)}
{\frac{{{\theta}(\widetilde{h}^{{(i)}}_y)^2 +
1}}{{(\widetilde{h}^{{(i)}}_y)^2 + {\theta}}}}\\
&=&\prod\limits_{y \in S(x)\cap V^2}
{\frac{{{\theta}(\widetilde{h}^{{(i)}}_y)^2 +
1}}{{(\widetilde{h}^{{(i)}}_y)^2 + {\theta}}}}\\
&=&{\left( {\frac{{\theta {(h^{(2)}_{i})^2} +
1}}{{{(h^{(2)}_{i})^2} + \theta }}} \right)^2}\\
&=&(h^{(2)}_{i})^2,
\end{eqnarray*}
here we used
$$
\prod\limits_{y \in S(x)\cap (V^k\setminus V^2)}
{\frac{{{\theta}(\widetilde{h}^{{(i)}}_y)^2 +
1}}{{(\widetilde{h}^{{(i)}}_y)^2 + {\theta}}}}=1.
$$

If $x\in V^k\setminus V^2$ then it is easy to see that
$S_k(x)\subset V^k\setminus V^2.$ Therefore we have
$$
(\widetilde{h}^{{(i)}}_x)^2 = \prod\limits_{y \in S(x)}
{\frac{{{\theta}(\widetilde{h}^{{(i)}}_y)^2 +
1}}{{(\widetilde{h}^{{(i)}}_y)^2 + {\theta}}}}=1.
$$

Thus $ \widetilde{h}^{{(i)}}_x, i=1,2$ satisfies the functional
equation \eqref{exp(h)} and we denote by $\mu_{\widetilde{h}^{{(i)}}_x},
i=1,2$ the Gibbs measures corresponding to $\widetilde{h}^{{(i)}}_x, i=1,2$ and those measures we called
$p$-adic ART generalized Gibbs measures. Thus, we have the following result:

\begin{thm} Let $p \equiv 1(\operatorname{mod }4)$. Then there exists at least three
$p$-adic generalized Gibbs measures for the Ising model on a Cayley tree order $k\geq3$.
\end{thm}

\textbf{Aknowledgements} The authors thank Professor U.A.Rozikov
for useful discussions.

\end{document}